\newif\ifarxiv
\arxivtrue

\ifarxiv 
\documentclass[sigconf, nonacm, authorversion]{aamas} 
\fi

\ifarxiv \else
\documentclass[sigconf]{aamas} 
\fi

\usepackage{algorithm}

\usepackage{enumerate}
\usepackage[shortlabels]{enumitem}

\usepackage{amsmath, amsthm}
\usepackage{algpseudocode}
\usepackage{multicol}
\usepackage{tikz}
\usetikzlibrary{shapes.misc, fit, decorations.pathreplacing,calligraphy,positioning} 
\usepackage{float}
\usepackage{harmony}
\usepackage{scalerel}
\usepackage{cleveref}
\usepackage{xcolor}
\usepackage{MnSymbol}
\usepackage{bm,xspace} 

\usepackage{soul}

\newcommand{\rust}[1]{{\color{black}#1}}
\newcommand{\maks}[1]{{\color{black}#1}}

\algtext*{EndIf}
\algtext*{EndFor}

\algnewcommand\algorithmiccase{\textbf{case}}
\algdef{SE}[CASE]{Case}{EndCase}[1]{\algorithmiccase\ #1}{\algorithmicend\ \algorithmiccase}%
\algtext*{EndCase}

\makeatletter
\newenvironment{breakablealgorithm}
  {
   \begin{center}
     \refstepcounter{algorithm}
     \hrule height.8pt depth0pt \kern2pt
     \renewcommand{\caption}[2][\relax]{
       {\raggedright\textbf{\fname@algorithm~\thealgorithm} ##2\par}%
       \ifx\relax##1\relax 
         \addcontentsline{loa}{algorithm}{\protect\numberline{\thealgorithm}##2}%
       \else 
         \addcontentsline{loa}{algorithm}{\protect\numberline{\thealgorithm}##1}%
       \fi
       \kern2pt\hrule\kern2pt
     }
  }{
     \kern2pt\hrule\relax
   \end{center}
  }
\makeatother

\newcommand{\lcirc}[1]{\mathchoice
{\mathbin{\vcenter{\hbox{\scaleobj{.75}{\mbox{\Kr{$#1$}}}}}}}
{\mathbin{\vcenter{\hbox{\scaleobj{.75}{\mbox{\Kr{$#1$}}}}}}}
{\mathbin{\vcenter{\hbox{\scaleobj{.5}{\mbox{\Kr{$#1$}}}}}}}
{\mathbin{\vcenter{\hbox{\scaleobj{.5}{\mbox{\Kr{$#1$}}}}}}}
}

\newcommand{\suggestion}[1]{{\color{blue}#1}}

\newcommand{\cL}{\mathcal{L}}

\newcommand{\lb}{\langle}
\newcommand{\rb}{\rangle}
\newcommand{\llb}{\lb \! \lb}
\newcommand{\rrb}{\rb \! \rb}

\newcommand{\X}{\mathsf{X}}

\newcommand{\U}{\mathsf{U}}
\newcommand{\R}{\mathsf{R}}
\newcommand{\G}{\mathsf{G}}

\newcommand{\F}{\mathsf{F}}

\newcommand{\AG}{\mathit{Agt}}
\newcommand{\AP}{\mathit{Prop}}
\newcommand{\Nom}{\mathit{Nom}}
\newcommand{\Act}{\mathit{Act}}

\newcommand{\ATL}{\mathsf{ATL}}

\newcommand{\LAMB}{\mathsf{LAMB}}
\newcommand{\HATL}{\mathsf{HATL}}

\protected\def\Ptime{\ifmmode \mbox{\sc P} \else {\sc P}\xspace\fi}
\protected\def\Pspace{\ifmmode \mbox{\sc Pspace} \else {\sc Pspace}\xspace\fi}
\protected\def\Exptime{\ifmmode \mbox{\sc Exptime} \else {\sc Exptime}\xspace\fi}
\protected\def\Dexptime{\ifmmode \mbox{\sc 2Exptime} \else {\sc 2Exptime}\xspace\fi}
\protected\def\NONELEMENTARY{\ifmmode \mbox{\sc NonElementary} \else {\sc NonElementary}\xspace\fi}
\protected\def\NP{\ifmmode \mbox{\sc NP} \else {\sc NP}\xspace\fi}

\newtheorem{remark}{Remark}

\usetikzlibrary{chains,fit,shapes}

\tikzset{%
round/.style={circle, draw=gray!60,fill=gray!5, very thick,minimum size=7mm, align=center},
dot/.style={draw, circle, minimum size=2mm,inner sep=0pt,outer sep=0pt,fill=black},
}

\usepackage{balance} 

\makeatletter
\gdef\@copyrightpermission{
  \begin{minipage}{0.2\columnwidth}
   \href{https://creativecommons.org/licenses/by/4.0/}{\includegraphics[width=0.90\textwidth]{by}}
  \end{minipage}\hfill
  \begin{minipage}{0.8\columnwidth}
   \href{https://creativecommons.org/licenses/by/4.0/}{This work is licensed under a Creative Commons Attribution International 4.0 License.}
  \end{minipage}
  \vspace{5pt}
}
\makeatother

\setcopyright{ifaamas}
\acmConference[AAMAS '25]{Proc.\@ of the 24th International Conference
on Autonomous Agents and Multiagent Systems (AAMAS 2025)}{May 19 -- 23, 2025}
{Detroit, Michigan, USA}{Y.~Vorobeychik, S.~Das, A.~Nowé  (eds.)}
\copyrightyear{2025}
\acmYear{2025}
\acmDOI{}
\acmPrice{}
\acmISBN{}

\acmSubmissionID{<<OpenReview submission id>>}

\title[Changing the Rules of the Game]{Changing the Rules of the Game:\\
Reasoning about Dynamic Phenomena in Multi-Agent Systems}\thanks{This is an extended version of the same title paper that will appear in AAMAS
2025. This version contains a technical appendix with proof details that, for space reasons, do
not appear in the AAMAS 2025 version. \textcolor{blue}{In this version we also correct the proof of Theorem 4.1 (the equivalence of $\mathsf{HATL}$ and $\mathsf{SLAMB}$). The previous proof was based on reduction axioms and it worked only for \emph{global} substitutions (like in $\mathsf{DEL}$ \cite{kooi07,kuijer14}). In this work, the substitutions are \emph{local}, and thus we should have reflected this in the proof. We provide a new proof of the claimed result. All new text related to the problem is in BLUE. We would like to thank St{\'e}phane Demri for spotting the problem and alerting us.}}

\author{Rustam Galimullin}
\affiliation{
  \institution{University of Bergen}
  \city{Bergen}
  \country{Norway}}
\email{rustam.galimullin@uib.no}

\author{Maksim Gladyshev}
\affiliation{
  \institution{Utrecht University}
  \city{Utrecht}
  \country{The Netherlands}}
\email{m.gladyshev@uu.nl}

\author{Munyque Mittelmann}
\affiliation{
  \institution{University of Naples Federico II}
  \city{Naples}
  \country{Italy}}
\email{munyque.mittelmann@unina.it}

\author{Nima Motamed}
\affiliation{
  \institution{Utrecht University}
  \city{Utrecht}
  \country{The Netherlands}}
\email{n.motamed@uu.nl}

\begin{abstract}
The design and application of multi-agent systems (MAS) require reasoning about the effects of modifications on their underlying structure.  
In particular, such changes may impact the satisfaction of system specifications and the strategic abilities of their autonomous components. 
In this paper, we are concerned with the problem of verifying and synthesising modifications (or \textit{updates}) of MAS. 
We propose an extension of the Alternating-Time Temporal Logic ($\ATL$) that enables reasoning about the dynamics of model change, called the \textit{Logic for $\ATL$ Model Building} ($\LAMB$). We show how $\LAMB$ can express various intuitions and ideas about the dynamics of MAS, from normative updates to mechanism design.
As the main technical result, we prove that, while being strictly more expressive than $\ATL$, $\LAMB$ enjoys a \Ptime-complete model-checking procedure.

\end{abstract}


\newcommand{\BibTeX}{\rm B\kern-.05em{\sc i\kern-.025em b}\kern-.08em\TeX}

\begin{document}


\pagestyle{fancy}
\fancyhead{}


\maketitle 


\section{Introduction}

Mechanism Design is a subfield of game theory  concerned with the design of  
mathematical structures (i.e. \textit{mechanisms}) 
describing the interaction of strategic agents 
that achieve 
desirable economic properties 
under the assumption of rational behavior 
\cite{Nisan2007}.  
Although it originated in economics, mechanism design provides an important foundation for the creation and analysis of multi-agent systems (MAS) \cite{dash2003computational,phelps2010evolutionary}. 
In numerous situations, the protocols and institutions describing interactions have already been designed and implemented. When those do not comply with the designer's objective (i.e. the economic properties) their complete redesign may not be feasible. For instance, a university with 
a selection procedure seen as  \textit{unfair}, 
would avoid implementing an entirely new procedure, but could be willing to adjust the existing one. 

Although logic-based approaches have been widely used for the verification \cite{Clarke2018}  and synthesis \cite{david2017program} of MAS, most research focuses on static or parametrised models and does not capture the dynamics of model change. One of the classic static approaches 
in the field is the Alternating-time Temporal Logic ($\ATL$) \cite{alur2002}.  $\ATL$ expresses conditions on the strategic abilities of agents interacting in a MAS, represented by a concurrent game model (CGM). 
In this paper, we are concerned with the problem of reasoning about the effects of modifications on CGMs. 
To tackle the problem, we extend $\ATL$ in two directions that have not been considered in the literature. 
First, we augment $\ATL$ with 
nominals and hybrid logic operators \cite{ARECES2007821}. The resulting logic is called Hybrid $\ATL$ ($\HATL$). $\HATL$ allows us to verify properties at states named by a given nominal, which cannot be captured by $\ATL$. Second, we propose the \textit{Logic for $\ATL$ Model Building} ($\LAMB$), which enhances $\HATL$ with update operators that describe explicit modular modifications in the model.

We define three fundamental update operators in $\LAMB$. First, we can \textit{change the valuation} of some propositional variable in a particular state to the valuation of a given formula. 
Second, we can \textit{switch the transition} from one state to another, 
which corresponds to modifying agents' abilities in a given state. 
Finally, we can \textit{add} a new state to the model and assign a fresh nominal to it. 
More complex operations, like adding a state, assigning it some propositional variable and adding incoming/outgoing transitions to it can 
be described in our language as sequences of primitive updates.  

Our intuitions on model updates are guided by research in Dynamic Epistemic Logic ($\mathsf{DEL}$) \cite{hvdetal.del:2007}, where one can reason about the effects of epistemic events on agents' knowledge. While epistemic and strategic reasoning are quite different domains, various $\mathsf{DEL}$s have been considered in the strategic setting (see, e.g., \cite{maubert20,delima14,agotnes08}). At the same time, research on updating models for strategic reasoning has been relatively sporadic and predominantly within the area of normative reasoning \cite{alechina18,alechina2022automatic}.


We deem our contribution to be two-fold. On the \textit{conceptual level}, we propose the exploration of general logic-based approaches to dynamic phenomena in MAS, not confined to particular implementation areas. 
On the \textit{technical level}, we propose a new 
formalism, $\LAMB$, to reason about modifications on MAS that
enables the verification of the strategic behavior of agents acting in a changing environment, and the synthesis of modifications on CGMs. 
We show how $\LAMB$ can express
various intuitions and ideas about the dynamics of MAS from
normative updates to mechanism design. We also formally study
the logic (and some of its fragments) from the point of view of expressivity and model-checking complexity. 
Our results show that $\LAMB$ is strictly more expressive than $\HATL$, which, in turn, is strictly more expressive than $\ATL$.  
Finally, we present a \Ptime-complete algorithm for model checking $\LAMB$.

\section{Reasoning About Strategic Abilities in the Changing Environment}

  
\subsection{Models}


Let $\AG = \{1,...,n\}$ be a non-empty finite set of agents. We will call subsets $C \subseteq \AG$ \textit{coalitions}, and complements $\overline{C}$ of $C$ \textit{anti-coalitions}. Sometimes we also call $\AG$ \textit{the grand coalition}. Moreover, let $\AP = \{p, q, ...\}$ and $\Nom = \{\alpha, \beta, ...\}$ be disjoint countably infinite sets of \textit{atomic propositions} and \textit{nominals} correspondingly. Finally, let $\Act = \{a_1, ..., a_m\}$ be a non-empty finite set of actions.

\begin{definition}[Named CGM]\label{def:cgs} \maks{Given a set of atomic propositions $\AP$, nominals $\Nom$ and agents $\AG$, a} \emph{named Concurrent Game Model} (nCGM) is a tuple $M = \lb S, \tau, L\rb$, where:
\begin{itemize}
    \item $S$ is a non-empty finite set of states;
    \item $\tau: 
    S\times \Act^{\AG} \to S$ is a transition function that assigns the outcome state $s' =\tau(s, (a_1,\dots, a_n))$ to a state $s$ and a tuple of actions $(a_1,\dots, a_n)$;
    \item $L: \AP \cup \Nom \to 2^{S}$ is a valuation function such that for all $\alpha \in \Nom: |L(\alpha)|\leqslant 1$ and for all $s \in S$, there is some $\alpha \in \Nom$ such that $L(\alpha) = \{s\}$.
\end{itemize}
We denote an nCGM $M$ with a \rust{designated} state $s$ as $M_s$.
Since all CGMs we are dealing with in this paper are named, we will abuse terminology and call nCGMs just CGMs.

Let $True(s)=\{p\in \AP \mid s\in L(p)\}\cup \{\alpha\in \Nom \mid s\in L(\alpha)\}$ be the set of all propositional variables and nominals that are true in state $s$. We define the \emph{size} of CGM $M$ as $|M| = |\AG| + |\Act| + |S| + |\tau| + \sum\limits_{s\in S}|True(s)|$, where $|\tau| = |S|\cdot |Act|^{|\AG|}$. We call a CGM \emph{finite}, if $|M|$ is finite. In this paper, we restrict our attention to finite models.

\end{definition}

Our models differ from standard CGMs in two ways. First, similarly to hybrid logic \cite{ARECES2007821}, states of our models have names represented by nominals $\Nom$. So, each nominal is assigned to at most one state, but each state may have multiple names. Observe that differently from hybrid logic, we allow nominals to have empty extensions, that is, to not be
assigned to any state. In the next section we will use this property to ensure that once new states are introduced to a model, we always have names available for them.  Finally, we also assume that our models are \textit{properly named}, i.e. each state is assigned some nominal. 



We also assume that all agents have the whole set $\Act$ of actions available to them. Such an assumption is relatively common in the strategy logics literature (see, e.g., \cite{mogavero10,aminof2019probabilistic,BOZZELLI2020199}), and allows for a clearer presentation of our framework.

\begin{definition}[Strategies] Given a coalition $C\subseteq \AG$, an \emph{action profile \rust{for coalition $C$}}, $A_C$,
is an element of $\Act^C$, and 
$\Act^{\AG}$ is the set of all \textit{complete} action profiles, i.e. all tuples $(a_1, \dots, a_n)$ for $\AG = \{1, \dots, n\}$. 

Given an action profile $A_C$, 
     we write $\tau(s, A_C)$ to denote a set $\{
     \tau (s,A)\mid A = A_C \cup A_{\overline{C}} , A_{\overline{C}}\in  Act^{\overline{C}} \}$.
     Intuitively, $\tau(s, A_C)$ is the set of all states reachable by (complete) action profiles that extend a given action profile $A_C$\rust{\footnote{\rust{Here, we here slightly abuse the notation and treat, whenever convenient, action profiles as sets rather than ordered tuples.}}}.
    
    A \emph{(memoryless) strategy profile} for $C$ is a function $\sigma_C: S \times C \to \Act$ with $\sigma_C(s,i)$ being an action agent $i$ takes in $s$.

    Given a CGM $M$, a \emph{play} $\lambda = s_0 s_1 \cdots$ is an infinite sequence of states in $S$ such that for all $i \geqslant 0$, state $s_{i+1}$ is a successor of $s_i$ i.e., there exists an action profile $A \in Act^{Agt}$ s.t. $\tau(s_i, A) = s_{i+1}$. We will denote the $i$-th element of play $\lambda$ by $\lambda[i]$. The set of all plays that can be realised by coalition $C$ following strategy $\sigma_C$ from some given state $s$, denoted by $\Lambda^s_{\sigma_C}$, 
    is defined as
    $$\{\lambda \mid \lambda[0] = s \text{ and } \forall i \in \mathbb{N}: \lambda[i+1] \in \tau(\lambda[i], \sigma_C (\lambda[i]))\}.$$
\end{definition}

\rust{
\begin{remark}
    Note that we consider only memoryless (positional) strategies here. One can also employ memory-full (perfect recall) strategies, where the choice of actions by a coalition depends not on the current state, but on the whole history of the system up to the given moment. We resort to the former for simplicity. Observe, however, that for $\ATL$, the semantics based on these two types of strategies are equivalent \cite{alur2002,schobbens04,jamroga15}.  
    We conjecture that this is also the case for $\LAMB$, and leave it for future work.
\end{remark}
}

\begin{example}
Examples of CGMs are given in Figure \ref{fig::exampleCGM}, where an arrow labelled, for example, by $ab$ denotes the action profile, where agent $1$ takes action $a$ and agent $2$ takes action $b$. 
In a state $s$ of model $M$,  the two agents can make a transition to state $t$, if they synchronise on their actions, i.e. if they choose the same actions. In model $N$, on the other hand, a similar transition can be performed whenever they choose  different actions. 
\begin{figure}[h!]
\centering
\scalebox{0.75}{
\begin{tikzpicture}
\node(-1) at (2,0) {$M$};
\node[circle,draw=black, minimum size=4pt,inner sep=0pt, fill = black, label=below:{$s:\alpha$}](1) at (0,0) {};
\node[circle,draw=black, minimum size=4pt,inner sep=0pt, , label=below:{$t:\beta$}](2) at (4,0) {};

\draw [->,thick](1) to [loop above] node[above, align=left] {$ab,ba$} (1);
\draw [->,thick](1) to [bend right] node[below,align=left] {$aa,bb$} (2);
\draw [->,thick] (2) to [bend right] node[above,align=left] {$aa,bb$} (1);
\draw [->,thick] (2) to [loop above] node[above,align=left] {$ab,ba$} (2);
\end{tikzpicture}
\hspace{10mm}
\begin{tikzpicture}
\node(-1) at (2,0) {$N$};
\node[circle,draw=black, minimum size=4pt,inner sep=0pt, fill = black, label=below:{$s:\alpha$}](1) at (0,0) {};
\node[circle,draw=black, minimum size=4pt,inner sep=0pt, , label=below:{$t:\beta$}](2) at (4,0) {};

\draw [->,thick] (1) to [loop above] node[above, align=left] {$aa,bb$} (1);
\draw [->,thick](1) to [bend right] node[below,align=left] {$ab,ba$} (2);
\draw [->,thick] (2) to [bend right] node[above,align=left] {$ab,ba$} (1);
\draw [->,thick] (2) to [loop above] node[above,align=left] {$aa,bb$} (2);
\end{tikzpicture}
}
\caption{CGMs $M$ and $N$ for two agents and two actions. Propositional variable $p$ is true in black states, and nominals $\alpha$ and $\beta$ are true in their corresponding states.}
\Description{CGMs $M$ and $N$ for two agents and two actions. Propositional variable $p$ is true in black states, and nominals $\alpha$ and $\beta$ are true in their corresponding states.}
\label{fig::exampleCGM}
\end{figure}
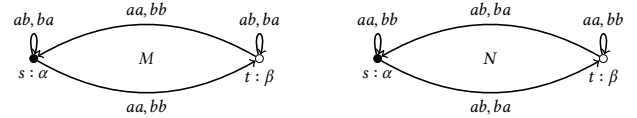 
\end{example}


\subsection{Hybrid $\ATL$}

We start 
with an extension of $\ATL$ with hybrid logic features that allow us to express properties at states named by a given nominal. 

\begin{definition}[Syntax of $\mathsf{HATL}$] The language of Hybrid ATL ($\mathsf{HATL}$) logic is defined recursively as follows
   \begin{alignat*}{3}
        &\mathsf{HATL} &&\thinspace \ni && \enspace \varphi ::= p \mid \alpha \mid @_\alpha \varphi \mid \neg \varphi \mid (\varphi \land \varphi) \mid \llb C \rrb \X \varphi \mid\\ 
        & && &&\mid \llb C \rrb \varphi \U \varphi \mid \llb C \rrb \varphi \mathsf{R} \varphi
\end{alignat*}
where $p \in \AP$, $\alpha \in \Nom$ and $C \subseteq \AG$. 
The fragment of $\mathsf{HATL}$ without $\alpha$ and $@_\alpha \varphi$ corresponds to $\ATL$ 
\footnote{We define $\mathsf{ATL}$ over $\{\llb C \rrb \X\varphi,$ $\llb C \rrb\varphi \U \varphi,$ $\llb C \rrb \varphi \mathsf{R} \varphi\}$  as it is strictly more expressive than $\mathsf{ATL}$ defined over $\{\llb C \rrb \X\varphi,$ $\llb C \rrb\varphi \U \varphi,$ $\llb C \rrb \G \varphi\}$ \cite{Laroussinie2008}.
}. 
\end{definition}

Here, $\llb C \rrb \X\varphi$ means that a coalition $C$ can ensure that $\varphi$ holds in the ne$\X$t state. The operator $\llb C \rrb\varphi \U \psi$ (here $\U$ stands for \textit{$\U$ntil}) means that $C$ has a strategy to enforce $\psi$ while maintaining the truth of $\varphi$. $\llb C \rrb \varphi \mathsf{R} \psi$ means that $C$ can maintain $\psi$ until $\varphi$ \textit{$\R$eleases} the requirement for the truth of $\psi$.  
Derived operators $\llb C \rrb\F\varphi =_{def} \llb C \rrb \top\U\varphi$ and $\llb C \rrb \G\varphi =_{def} \llb C \rrb\bot\R\varphi$ mean ``$C$ has a strategy to \textit{eventually} make $\varphi$ true" and ``$C$ has a strategy to make sure that $\varphi$ is \textit{always} true" respectively. Hybrid logic operator $@_\alpha \varphi$ means ``at state named $\alpha$, $\varphi$ is true". This operator allows us to `switch' our point of evaluation to the state labelled with nominal $\alpha$ in the syntax. 
Given a formula $\varphi\in \mathsf{HATL}$, the \emph{size of} $\varphi$, denoted by $|\varphi|$, is the number of symbols in $\varphi$.

\begin{definition}[Semantics of $\mathsf{HATL}$]\label{def:semantics_static}
Let $M = \lb S, \tau, L\rb$ be a CGM, $s\in S$, $p \in \AP$, $\alpha \in \Nom$, and $\varphi, \psi \in \mathsf{HATL}$. The semantics of $\mathsf{HATL}$  is defined by induction as follows: 
    \begin{alignat*}{3}
        &M_s \models p && \text{ iff } && s \in L(p)\\
        &M_s \models \alpha && \text{ iff } && s \in L(\alpha)\\
        &M_s \models @_\alpha \varphi && \text{ iff } && \exists t \in S: \{t\} = L(\alpha)  \text{ and } M_t \models \varphi\\
        &M_s \models \lnot \varphi && \text{ iff } && M_s \not \models \varphi\\
        &M_s \models\varphi \land \psi && \text{ iff } && M_s \models \varphi \text{ and } M_s \models \psi\\
        &M_s \models \langle \! \langle C \rangle \! \rangle \mathsf{X} \varphi && \text{ iff } && \exists \sigma_C, \forall \lambda \in \Lambda^s_{\sigma_C}: M_{\lambda[1]} \models \varphi\\
        &M_s \models \langle \! \langle C \rangle \! \rangle \psi \mathsf{U} \varphi && \text{ iff } && \exists \sigma_C, \forall \lambda \in \Lambda^s_{\sigma_C}, \exists i \geqslant 0: M_{\lambda[i]} \models \varphi  \\
        & && && \text{ and } M_{\lambda[j]} \models \psi
        \text{ for all } 0 \leqslant j < i\\  
        &M_s \models \langle \! \langle C \rangle \! \rangle \psi \mathsf{R} \varphi && \text{ iff } &&   \exists \sigma_C, \forall \lambda \in \Lambda^s_{\sigma_C}, 
        \forall i \geqslant 0: M_{\lambda[i]} \models \varphi \\
        & && && \text{ or } 
        M_{\lambda[j]} \models \psi  \text{ for some } 0 \leqslant j \leqslant i       
\end{alignat*}     
\end{definition}

Observe that differently from standard hybrid logics \cite{ARECES2007821}, the truth condition for $@_\alpha \varphi$ also requires the state with name $\alpha$ to exist. This modification is necessary as we let some nominals to have empty denotations. 

\begin{example}
Recall CGM $M$ from Figure \ref{fig::exampleCGM}. It is easy to see that 
$M_s \models \alpha$ meaning that state $s$ is named $\alpha$; $M_s \models \llb \{1,2\} \rrb \X \lnot p$, i.e. that from state $s$ the grand coalition can force $\lnot p$ to hold in the next state; and $M_s \models @_\beta \llb \{1,2\} \rrb \G \beta$ meaning that in the state named $\beta$ 
the grand coalition has a strategy to always remain in this state.
\end{example}


\paragraph*{Expressivity of $\mathsf{HATL}$.} 


Even though the interplay between nominals and temporal modalities has been quite extensively studied (see, e.g., \cite{blackburn99,Goranko2000-GORCTL-2,franceschet06,
Kernberger2020}), to the best of our knowledge, the extension of $\ATL$ with nominals has never been considered. 
Our \textit{Hybrid} $\ATL$ provides us with extra expressive power, compared to the standard $\ATL$. 
We can, for instance, formulate safety and liveness properties in terms of \textit{names}.
For example, 
formula
\[\llb C \rrb \G \, \mathsf{safe} \wedge  \llb D \rrb  \neg \mathsf{crashed}\,\U\, {\alpha} \wedge \llb \emptyset \rrb \F \beta\]
states that (1) coalition $C$ can enforce that only $\mathsf{safe}$ states will be visited, 
(2) coalition $D$ can 
avoid crashing until a state named 
$\alpha$ is visited, and (3) no matter what the agents $\AG$ do, a state named ${\beta}$ will eventually be visited. Both conjuncts (2) and (3) use nominals to refer to state names in syntax, and are not expressible in $\ATL$. 
Given that (as we will show later) the valuation of propositions in dynamic systems may change over time, as well as the transitions affecting reachability of some states, we believe that 
$\mathsf{HATL}$ allows for a more fine-grained way to capture safety and liveness. 

\begin{definition}
Let $\mathsf{L}_1$ and $\mathsf{L}_2$ be two languages, and let $\varphi \in \mathsf{L}_1$ and $\psi \in \mathsf{L}_2$. 
We say that $\varphi$ and $\psi$ are \emph{equivalent}, when for all CGMs $M_s$: $M_s \models \varphi$ if and only if $M_s \models \psi$.

If for every $\varphi \in \mathsf{L}_1$ there is an equivalent $\psi \in \mathsf{L}_2$, we write $\mathsf{L}_1 \preccurlyeq \mathsf{L}_2$ and say that $\mathsf{L}_2$ is \emph{at least as expressive as} $\mathsf{L}_1$. We write $\mathsf{L}_1 \prec \mathsf{L}_2$ iff $\mathsf{L}_1 \preccurlyeq \mathsf{L}_2$ and $\mathsf{L}_2 \not \preccurlyeq \mathsf{L}_1$, and we say that $\mathsf{L}_2$ is \emph{strictly more expressive than} $\mathsf{L}_1$. Finally, if $\mathsf{L}_1 \preccurlyeq \mathsf{L}_2$ and $\mathsf{L}_2 \preccurlyeq \mathsf{L}_1$, we say that $\mathsf{L}_1$ and $\mathsf{L}_2$ are \emph{equally expressive} and write $\mathsf{L}_1 \approx \mathsf{L}_2$.
\end{definition}

\rust{Let us return to \Cref{fig::exampleCGM}, and assume that we have model $N'$, which is exactly like $N$ with the only difference that $N'_s \models q$.} Now, CGMs $M$ and $N'$ can be viewed as a \maks{disjoint union} \cite{thebluebible} $M \biguplus N'$ (modulo renaming states in $N'$ and assigning to them other nominals than $\alpha$ and $\beta$, \rust{like $\gamma$ and $\delta$ correspondingly})
of two isolated submodels $M$ and $N'$. 
Note that no $\ATL$ formula $\varphi$ that holds in $M$ depends on the submodel $N'$ as there are no transitions there. Hence, $\ATL$ cannot distinguish between $M$ and $M \biguplus N'$.
At the same time, \maks{in contrast to state labels, we can use nominals in the syntax. So, formulas of $\mathsf{HATL}$ containing $@_\alpha$ operators can access states in the submodel $N'$, and hence can have different truth values in $M$ and $M \biguplus N'$. An example of such a formula would be $@_\gamma q$, which is false everywhere in $M$ (since there is no state named $\gamma$), and true everywhere in $M \biguplus N'$.} This trivially implies that $\mathsf{HATL}$ \textit{is strictly more expressive than} $\mathsf{ATL}$ ($\mathsf{ATL}\prec \mathsf{HATL}$).

\subsection{Logic for $\ATL$ Model Building}

Now we turn to the full logic with dynamic update operators. 



\begin{definition}[Syntax of $\LAMB$]
The language $\LAMB$ of \emph{Logic for ATL Model Building} is defined 
as follows 
\begin{alignat*}{4}
        &\LAMB &&\thinspace \ni && \enspace \varphi ::= &&p \mid \alpha \mid @_\alpha \varphi \mid \neg \varphi \mid \varphi \land \varphi \mid \llb C \rrb \X \varphi \mid \\
        & && && &&   \mid \llb C \rrb \varphi \U \varphi \mid \llb C \rrb \varphi \mathsf{R} \varphi \mid [\pi]\varphi\\
        & &&\thinspace  && \enspace \pi ::= &&(p_\alpha:=\psi) \mid \alpha \xrightarrow{A} \alpha \mid \lcirc{\alpha}  
\end{alignat*}
where $p \in \AP$, $\alpha \in \Nom$, $C \subseteq \AG$, $A \in \Act^{\AG}$. 
We will also write $[\pi_1; ...; \pi_n]\varphi$ for $[\pi_1]...[\pi_n]\varphi$.  
Two important fragments of $\LAMB$ that we will also consider are the one without constructs $\alpha \xrightarrow{A} \alpha$ and $\lcirc{\alpha}$, called \emph{substitution} $\LAMB$ ($\mathsf{S}\LAMB$); and the one without $p_\alpha:=\psi$  and $\lcirc{\alpha}$, called \emph{arrow} $\LAMB$ ($\mathsf{A}\LAMB$).
\end{definition}

\begin{definition}[Semantics of $\LAMB$]
\label{def:semLAMB}
Let $M = \lb S, \tau, L\rb$ be a CGM, $s\in S$, $p \in \AP$, $\alpha \in \Nom$, and $\varphi, \psi \in \LAMB$. The semantics of $\LAMB$  is defined as in Definition \ref{def:semantics_static} with the following additional cases: 
    \begin{alignat*}{3}
        &M_s \models [p_\alpha:=\psi]\varphi && \quad \text{ iff } \quad && M^{p_\alpha:=\psi}_s \models \varphi\\
        &M_s \models [\lcirc{\alpha}]\varphi && \quad \text{ iff } \quad && M^{\lcirc{\alpha}}_s \models \varphi\\
        &M_s \models [\alpha \xrightarrow{A} \beta]\varphi && \quad \text{ iff } \quad && M^{\alpha \xrightarrow{A} \beta}_s \models \varphi
\end{alignat*} 
where $M^\pi$ with $\pi \in \{p_\alpha:=\psi, \alpha \xrightarrow{A} \beta, \lcirc{\alpha}\}$ is called an \emph{updated} CGM, and is defined as follows:
\begin{itemize}
    \item $M^{p_\alpha:=\psi}_s = \lb S, \tau, L^{p_\alpha:=\psi}\rb$, where, if $\exists t \in S$ such that $L(\alpha) = \{t\}$, then  
    \[ L^{p_\alpha:=\psi } (p) =
    \begin{cases}

        L(p) \cup \{t\} &\text{if } 
        M_s \models \psi,\\
        L(p) \setminus \{t\} &\text{if } 
        M_s \not \models \psi,\\

    \end{cases}\]
    and $M^{p_\alpha:=\psi} = \lb S, \tau, L\rb$ otherwise.
    \item  $M^{\alpha \xrightarrow{A} \beta} =  \lb S, \tau^{\alpha \xrightarrow{A} \beta}, L\rb$, where, if $\exists s,t \in S$ such that $L(\alpha) = s$ and $L(\beta) = t$, then
    \[\tau^{\alpha \xrightarrow{A} \beta}(s',A') =
        \begin{cases}
            t, &\text{if } s' = s \text{ and } A' = A\\
            \tau(s',A') &\text{otherwise},\\
        \end{cases}
    \]
    and  $ M^{\alpha \xrightarrow{A} \beta} =  \lb S, \tau, L\rb$ otherwise.
    \item $M^{\lcirc{\alpha}} = \lb S^{\lcirc{\alpha}}, \tau^{\lcirc{\alpha}}, L^{\lcirc{\alpha}}\rb$, where, if $L(\alpha) = \emptyset$, then  $S^{\lcirc{\alpha}} = S \cup \{t\}$ with  $t \not \in S$, $\tau^{\lcirc{\alpha}} = \tau \cup \{(t, A, t) \mid \forall A \in \Act^\AG\}$, and $L^{\lcirc{\alpha}}(\alpha) = L \cup \{(\alpha, \{t\})\}$. If $L(\alpha) \neq \emptyset$, then $M^{\lcirc{\alpha}} = \lb S, \tau, L\rb$. 
\end{itemize}
\end{definition}

Intuitively, updating a given model with $p_\alpha := \psi$ assigns to propositional variable $p$ in the state named $\alpha$ (if there is such a state) the truth value of $\psi$ in the state of evaluation. Observe, however, that due to the presence of operators $@_\alpha$, we can also let the truth-value of $p$ be dependent on the truth of $\psi$ in any other state. Updating with $\alpha \xrightarrow{A} \beta$ results in redirecting the $A$-labelled arrow that starts in state named $\alpha$ from some state $\tau(L(\alpha), A)$ to state named $\beta$ (if states named $\alpha$ and $\beta$ exist). Finally, operator $\lcirc{\alpha}$ adds a new state to the model and gives it name $\alpha$. All propositional variables are false in the new state and all transitions are self-loops. The self-loops and the fact that propositional variables are false in the new state is a design choice that achieves two goals. First, adding new state results in a finite model (recall that to determine the size of a model we count all true propositions), and, second, we deem this to be the smallest meaningful change that is in line with the idea of modularity. All valuations and self-loops can be then further modified by the corresponding substitutions and arrow operators.   

\begin{example}
\label{ex:expressive}
    Consider CGMs $M$ and $N$ from Figure \ref{fig::exampleCGM} and an update $\alpha \xrightarrow{aa} \alpha$. Observe that such an update leaves $N$ intact as there is already a self-loop labelled $aa$ in the $\alpha$-state. Model $M$, on the other hand, is changed and the resulting updated model $M^{\alpha \xrightarrow{aa} \alpha}$  is depicted in Figure \ref{fig::updateExample}. 

    \begin{figure}[h!]
\centering
\scalebox{0.75}{
\begin{tikzpicture}
\node(-1) at (2,0) {$M^{\alpha \xrightarrow{aa} \alpha}$};
\node[circle,draw=black, minimum size=4pt,inner sep=0pt, fill = black, label=below:{$s:\alpha$}](1) at (0,0) {};
\node[circle,draw=black, minimum size=4pt,inner sep=0pt, , label=below:{$t:\beta$}](2) at (4,0) {};

\draw[->,thick] (1) to [loop above] node[above, align=left] {$ab, ba, aa$} (1);
\draw [->,thick](1) to [bend right] node[below,align=left] {$bb$} (2);
\draw [->,thick] (2) to [bend right] node[above,align=left] {$aa, bb$} (1);
\draw [->,thick] (2) to [loop above] node[above,align=left] {$ab, ba$} (2);

\node(-1) at (2,-3.3) {$M^{\pi_1}$};
\node[circle,draw=black, minimum size=4pt,inner sep=0pt, fill = black, label=below:{$s:\alpha$}](1) at (0,-3.3) {};
\node[circle,draw=black, minimum size=4pt,inner sep=0pt, , label=below:{$t:\beta$}](2) at (4,-3.3) {};

\node[rectangle,draw=black, minimum size=4pt,inner sep=0pt, fill = black, label=below:{$u:\gamma$}](3) at (0,-5.3) {};

\draw [->,thick] (1) to [loop above] node[above, align=left] {$ab, ba$} (1);
\draw [->,thick](1) to [bend right] node[below,align=left] {$aa,bb$} (2);
\draw [->,thick] (2) to [bend right] node[above,align=left] {$aa,bb$} (1);
\draw [->,thick] (3) to [loop above] node[above,align=left] {$aa,ab$\\$ba,bb$} (3);

\end{tikzpicture}
\begin{tikzpicture}
\node(-1) at (2,0) {$M^{\pi_2}$};
\node[circle,draw=black, minimum size=4pt,inner sep=0pt, fill = black, label=above:{$s:\alpha$}](1) at (0,0) {};
\node[circle,draw=black, minimum size=4pt,inner sep=0pt, , label=above:{$t:\beta$}](2) at (4,0) {};

\node[rectangle,draw=black, minimum size=4pt,inner sep=0pt, fill = black, label=below:{$u:\gamma$}](3) at (0,-2) {};

\draw [->,thick](1) to [bend right] node[below,align=left] {$aa, bb$} (2);
\draw [->,thick] (2) to [bend right] node[above,align=left] {$aa, bb$} (1);
\draw [->,thick](1) to [bend right] node[left,align=left] {$ab$\\$ba$} (3);
\draw [->,thick] (3) to [loop above] node[above right,align=left] {$ab$\\$ba$} (3);
\draw [->,thick](3) to[bend right] node[below,align=left, near end] {$aa$\\$bb$} (2);

\node(-1) at (2,-3.3) {$M^{\mathit{SN}}$};
\node[circle,draw=black, minimum size=4pt,inner sep=0pt, fill = black, label=above:{$s:\alpha$}](1) at (0,-3.3) {};
\node[circle,draw=black, minimum size=4pt,inner sep=0pt, , label=above:{$t:\beta$}](2) at (4,-3.3) {};

\node[rectangle,draw=black, minimum size=4pt,inner sep=0pt, fill = black, label=below:{$u:\gamma$}](3) at (0,-5.3) {};

\node[rectangle,draw=black, minimum size=4pt,inner sep=0pt, label=below:{$v:\delta$}](4) at (4,-5.3) {};

\draw [->,thick](1) to [bend right] node[below,align=left] {$aa,bb$} (2);
\draw [->,thick] (2) to [bend right] node[above,align=left] {$aa,bb$} (1);
\draw [->,thick](1) to [bend right] node[left,align=left] {$ab$\\$ba$} (3);
\draw [->,thick] (3) to [loop above] node[above,align=left] {$ab$\\$ba$} (3);
\draw [->,thick](2) to [bend left] node[right,align=left] {$ab$\\$ba$} (4);
\draw [->,thick] (4) to [loop above] node[above,align=left] {$ab$\\$ba$} (4);
\draw [->,thick](3) to[bend right] node[below,align=left, near start] {$aa$\\$bb$} (2);
\draw [->,thick](4) to[bend left] node[below,align=left, near start] {$aa$\\$bb$} (1);
\end{tikzpicture}

}
\caption{Updated CGMs $M^{\alpha \xrightarrow{aa} \alpha}$ (top left), \rust{$M^{\pi_1
}$ (bottom left), $M^{\pi_2}$ (top right),} and $M^{\mathit{SN}}$ (bottom right). Proposition $p$ is true in black states, and $\mathit{fine}$ is true in square states.}
\Description{Updated CGMs $M^{\alpha \xrightarrow{aa} \alpha}$ (top left), \rust{$M^{\pi_1
}$ (bottom left), $M^{\pi_2}$ (top right),} and $M^{\mathit{SN}}$ (bottom right). Proposition $p$ is true in black states, and $\mathit{fine}$ is true in square states.}
\label{fig::updateExample}
\end{figure}
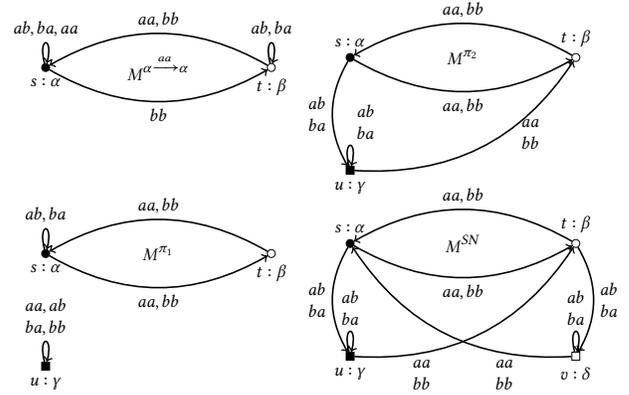 

Now it is easy to see, for example, that in the updated model the first agent can force the system to stay in the $\alpha$-state by choosing action $a$, i.e. $M_s \models [\alpha \xrightarrow{aa} \alpha] \llb \{1\} \rrb \X \alpha$, while it is not the case for $N^{\alpha \xrightarrow{aa} \alpha}$ (which is the same as $N$), i.e.  $N_s \not \models [\alpha \xrightarrow{aa} \alpha] \llb \{1\} \rrb \X \alpha$. 

\rust{
As a more complex update, consider $\pi_1:= \lcirc{\gamma}; p_\gamma:=\top; \textit{f}ine_\gamma:=\top$ and the resulting model $M^{\pi_1}$, where we have a new state named $\gamma$ that satisfies propositions $p$ and $\textit{f}ine$ (the intuition behind $\textit{f}ine$ is given in Section \ref{sec:norms}). Further, we can add redirection of some edges $\pi_2:= \pi_1; \alpha \xrightarrow{ab} \gamma; \alpha \xrightarrow{ba} \gamma; \gamma \xrightarrow{aa} \beta; \gamma \xrightarrow{bb} \beta$ to obtain model $M^{\pi_2}$.
The final complex update of $M$ with $\mathit{SN}$, also presented in Figure \ref{fig::updateExample}, is discussed in the next section in the context of normative updates.
}

\end{example}


\section{Dynamic MAS through the lens of $\LAMB$}

In this section, we show how our dynamic approach to CGMs allows us to capture various ideas and intuitions from the MAS research. 







\subsection{Normative Updates on CGMs} 
\label{sec:norms}

Updates of CGMs play a crucial role in the area of normative MAS (see \cite{alechina18} for an overview). In the context of $\ATL$, the general idea is to divide actions of agent's into those that are permitted to be executed in a current state, and into those that are not. We argue that $\LAMB$, thanks to its dynamic features, captures various intuitions and technicalities of reasoning about norms in MAS.

As a use case of $\LAMB$, we consider \textit{norm-based updates} as presented in \cite{bulling2016norm}.\footnote{We simplify the exposition for the sake of clarity.} The authors distinguish two types of norms for CGMs: regimenting norms and sanctioning norms. 
Regimenting norms prohibit certain transitions by looping them to the origin state. This corresponds to the intuition that a selected action profile has no noticeable effect and thus the system stays in the same state. 
Note that agents in this case always have available actions. This  assumption is sometimes called \textit{reasonableness} 
\cite{agotnes_07,hoek_07,agotnes10,agotnes10b,alechina13,galimullin2024synthesizing}. 

While it is quite straightforward to model regimenting norms in $\LAMB$ using arrow updates, for the sake of an example we focus on more general and subtle sanctioning norms. Such norms put sanctions, or fines, on certain action profiles without explicitly prohibiting them. A bit more formally, a sanctioning norm is a triple $(\varphi, \mathcal{A}, \mathcal{S})$, where $\varphi$ is a formula of a logic, $\mathcal{A}$ is a set of action profiles, and  $\mathcal{S}$ is a set of sanctions, which is a subset of the set of propositional variables. Intuitively, sanctioning norm $(\varphi, \mathcal{A}, \mathcal{S})$ states that performing action profiles from $\mathcal{A}$ from states satisfying $\varphi$ imposes sanctions from $\mathcal{S}$. 


The result of updating a CGM $M$ with a norm $\mathsf{SN} = (\varphi, \mathcal{A}, \mathcal{S})$ is a CGM $M^\mathsf{SN}$ such that for all $\mathcal{A}$-transitions from $\varphi$-states, we create copies of those states and make sanctioning atoms $\mathcal{S}$ true in those copies. Non-sanctioned transitions from the copy-states have the same outcome as the original transitions in the initial model. 

Consider as an example the sanctioning norm $\mathsf{SN} = (\top, \{ab,ba\}, \allowbreak \{\mathit{fine}\})$ and CGM $M$ from Example \ref{fig::exampleCGM}. Recall that in $M$, agents can switch the current state only if they cooperate (i.e. choose the same actions). Hence, norm $\mathsf{SN}$ penalises non-cooperative behaviour with the sanction $\mathit{fine}$. 
We can directly model the implementation of $\mathsf{SN}$ on $M$ in $\LAMB$ by translating $\mathsf{SN}$ into a complex update
\begin{alignat*}{1}
    &\mathit{SN}=\lcirc{\gamma}; p_\gamma:=\top; \mathit{fine}_\gamma:= \top; \alpha \xrightarrow{ab} \gamma; \alpha \xrightarrow{ba} \gamma; \gamma \xrightarrow{aa} \beta;\\
    &\gamma \xrightarrow{bb} \beta; \lcirc{\delta}; \mathit{fine}_\delta:= \top; \beta \xrightarrow{ab} \delta; \beta \xrightarrow{ba} \delta; \delta \xrightarrow{aa} \alpha;\delta \xrightarrow{bb} \alpha.
\end{alignat*}
The result of updating $M$ with $\mathit{SN}$ \rust{as well as some intermediate updates} is presented in Figure \ref{fig::updateExample}. In $M^{\mathit{SN}}$, we create copies of states $s$ and $t$ (both satisfying $\top$), named $u$ and $v$, which now also satisfy sanction $\mathit{fine}$. Then, undesirable action profiles $\{ab,ba\}$ originating in $s$ and $t$ lead to the sanctioned states $u$ and $v$. At the same time, action profiles not from $\{ab,ba\}$ behave similarly to the initial model. It is easy to see that our $\LAMB$ update $\mathit{SN}$ on $M$ captures the effect of $\mathsf{SN}$ on $M$. 



\subsection{Synthesis} 
\label{sec:synth}
\rust{Assume that you have a model $M_s$ of a MAS that does not satisfy some safety requirement $\varphi$. One way to go about it would be to create a completely new model from scratch, or to try to manually fix the existing one. However, in case of large models and complex safety requirements, both options may not be feasible.
This is exactly where the problem of \textit{synthesis} (or \textit{modification synthesis}) comes in. 
The full language of $\LAMB$ is expressive enough to capture the (bounded) synthesis problem from a given specification and starting model. In such a way, the modifications 
of the model are presented as updates operators of $\LAMB$. }

\rust{
\begin{definition}
    For a given CGM $M_s$, formula $\varphi$, and natural number $n$, the \emph{bounded (modification) synthesis existence problem} decides whether there is an update $\pi_\varphi:= [\pi_1, ..., \pi_k]$, with the size at most $n$, such that $M_s \models [\pi_\varphi] \varphi$. 
\end{definition}

The bounded synthesis existence problem is important, as it tackles the synthesis of \textit{compact} modifications. Indeed, oftentimes designers of (a model of) MAS are interested in a modification that achieves the goal without significantly altering the initial model. Think of an example of a software update, where we would like to extend the functionality of the software without rewriting most of its code. Moreover, the bounded 
synthesis 
problem is also important in  cases where changes in the given model 
are costly, and we want to avoid wasting resources as much as possible. 

In Section \ref{sec:mc}, we show that the complexity of the model checking problem for $\LAMB$ is \Ptime-complete. With this in mind, it is easy to see that the complexity of the bounded synthesis existence problem is \NP-complete. 

\begin{proposition}
The bounded synthesis problem for $\LAMB$ is \NP-complete.
\end{proposition}

\begin{proof}
    To see that the problem is in \NP, let $M_s$ be a CGM, $\varphi$ be a formula, and $n$ be a natural number. We guess a $\pi_\varphi$ of size at most $n$. It follows that $[\pi_\varphi] \varphi$ is of polynomial size 
    w.r.t $|\varphi| + n$. Then we can model check $M_s \models [\pi_\varphi] \varphi$ in polynomial time (\Cref{thm:MC}).

    For the \NP-hardness, we employ the reduction from 3-SAT problem. Let $\varphi:= \bigwedge_{1 \leqslant i \leqslant k} (\psi_{i,1} \lor \psi_{i,2} \lor \psi_{i,3})$, where $\psi$'s are literals, be an instance of 3-SAT, and let $P^\varphi = \{p^1, ..., p^m\}$ be the set of propositions appearing in $\varphi$. We construct model $M^\varphi$ over one agent consisting of a single state $s$ with the name $\alpha$, and such that $V^\varphi (p^i) = \emptyset$ for all $p^i \in P^\varphi$. All transitions for the agent are self-loops. Now it is easy to see that $\varphi$ is satisfiable iff there is a $\pi_\varphi$ consisting only of substitutions $p^i_\alpha:= \top$, and thus of the size linear in $|P^\varphi|$, such that $M_s^\varphi \models [\pi_\varphi]\varphi$. In other words, if there is an assignment that makes $\varphi$ true, we can explicitly simulate it in $M_s^\varphi$ using constructs $p^i_\alpha:= \top$ for variables $p^i \in P^\varphi$ that should be assigned `true'. 
\end{proof}


The fact that $\LAMB$ captures the modification synthesis 
is significant, because a \textit{constructive} solution to the problem would produce a step-by-step recipe, or an instruction, of how to modify a given model to make it satisfy some desirable property $\varphi$. 
In this section we studied the computational complexity of checking \textit{whether} a required modification of certain size exists. In future work, we will focus on constructive solutions to the problem, i.e. providing algorithms that automatically construct the required $\LAMB$ update. }

\subsection{Mechanisms for Social Choice }

\newcommand{\egdef}{:=}


The key advantage of synthesising mechanisms from logical specifications is that, as a declarative approach, the designer is not required to construct a complete solution for the problem of interest; instead, she can describe the desired mechanism in terms of its rules and desirable properties. Some recent approaches use model checking and satisfiability procedures for Strategy Logic \cite{SLKF_KR21,MittelmannMMP22}. We, however, focus on the dynamic $\LAMB$ approach. 

Let us assume a mechanism encoded as a CGM, and a finite set of alternatives $Alt$. Such a mechanism may represent, for instance, a single-winner election or a resource allocation protocol.
In such CGM, 
we let the atomic  proposition $\textit{pref(i,j)}_a$ denote that agent $a$ prefers the alternative $i$ to $j$ (e.g., she prefers that $i$ is elected over $j$).
We also let $\textit{dislike}_a(i)$ indicate that $i$ is \textit{disliked} by agent $a$, 
and $chosen_i$ denote that the alternative $i$ was chosen. 

We now illustrate how to capture two classic concepts from game theory: individual rationality and Pareto optimality.  
Individual rationality expresses the idea that each agent can ensure nonnegative utility \cite{Nisan2007}. In our setting, this can be seen as avoiding disliked candidates and expressed with the following formula  
$$ \bigwedge_{a \in Agt, i \in Alt} \llb \{a\} \rrb \G \neg ( \textit{chosen}_i \land \textit{disliked}_a(i))$$
which states that each agent has a strategy to enforce that none of the disliked alternatives are ever chosen. 

A mechanism is  Pareto optimal if any change of outcome that is beneficial to one agent is detrimental to at least one of the other agents.
The formula
$$pareto_i \egdef  \bigwedge_{a \in Agt,  j \in Alt\setminus\{i\}}\big(
( \textit{pref}_a(j,i)
 \rightarrow 
\bigvee_{b \in Agt} \textit{pref}_b(i,j)
   \big)$$
expresses that  an alternative $i$ is Pareto optimal  whenever for any other alternative $j$, if $j$ is preferred to an agent $a$, then there is an agent $b$ that prefers $i$ over $j$. 

We can then express that there is a strategy for the coalition $C$ to ensure that any chosen alternative is Pareto optimal, with the  formula $ \varphi_{po}:=  \llb C \rrb  \G \bigwedge_{i \in Alt} chosen_i \to pareto_i
$. 
If we let $C = \emptyset$, this formula requires choices to be Pareto optimal for any possible behavior of the agents.

Once desirable mechanism properties are expressed as $\ATL$ formulas, one can use $\LAMB$ to verify whether CGM $M_s$ has such property or would have it in case a sequence of modifications $\pi_1; ..., \pi_n$ was performed, i.e. whether $M_s \models [\pi_1; ...; \pi_n] \varphi_{po}$ for Pareto optimality. Further, 
\rust{
employing the ideas of synthesis, one could automatically obtain the required modifications. 
}

\section{Expressivity}

In this section, we compare the expressive power of $\LAMB$ and its fragments. In particular, we show that, interestingly enough, substitutions on their own do not add expressive power compared to the base $\HATL$ (Theorem \ref{thm:hatlEQslamb}). The ability to move arrows, on the other hand, leads to an increase in expressivity (Theorem \ref{thm:hatlVSalamb}). 


First, we show that $\mathsf{SLAMB}$ and $\mathsf{HATL}$ are equally expressive. This is quite intriguing since it means that adding substitutions to the base logic $\mathsf{HATL}$ does not allow us to express anything that we could not express in $\mathsf{HATL}$.
We prove this 
result by providing a truth-preserving translation from formulas of $\mathsf{S}\LAMB$ into formulas of $\mathsf{H}\ATL$. \textcolor{blue}{The expressivity results that use translation schemas (usually based on \textit{reduction axioms}) are quite ubiquitous in $\mathsf{DEL}$ (see \cite{hvdetal.del:2007} for an overview, and 
\cite{kooi07,kuijer14} for substitution specific translations).}


\begin{theorem}
\label{thm:hatlEQslamb}
$\mathsf{HATL}\approx \mathsf{SLAMB}$
\end{theorem}

\begin{proof}

{\color{blue}First, we show how to `translate away' a substitution that has in its scope only formulas of $\mathsf{HATL}$. Specifically, let $M_s$ be a model, and $\varphi \in \mathsf{HATL}$. Then, $M_s \models [p_\alpha := \psi]\varphi$ if and only if $M_s \models (\lnot @_\alpha \top \rightarrow \varphi) \land (@_\alpha \top \rightarrow ((\psi \leftrightarrow @_\alpha p) \to \varphi) \land (\lnot(\psi \leftrightarrow @_\alpha p) \to \varphi[p / \alpha \to \lnot p \land \lnot \alpha \to p]))$.

Recall, that by the definition of semantics, $M_s^{p_\alpha := \psi}$ is similar to $M_s$ with the only difference that if $\exists t \in S: t = L(\alpha)$, then we consider two options. If $M_s \models \psi$, then $L^{p_\alpha := \psi} (p) = L(p) \cup \{t\}$, and if $M_s \models \lnot \psi$, then $L^{p_\alpha := \psi} (p) = L(p) \setminus \{t\}$. Otherwise, $M_s^{p_\alpha := \psi} = M_s$.

The translation explicitly captures these semantic conditions. Hence, given an arbitrary $M_s$, consider $M_s \models [p_\alpha := \psi]\varphi$ with $\varphi \in \mathsf{HATL}$. If no state in $M$ is named $\alpha$, we evaluate $\varphi$ in $s$. This corresponds to the conjunct $\lnot @_\alpha \top \to \varphi$ in the translation.

    Now, assume that there is a state named $\alpha$ in the model, i.e. that $M_s \models @_\alpha \top$. From the definition of the semantics of the substitution, we have that $L^{p_\alpha := \psi} (p)$ is updated iff $M_s \models \psi$ is not equivalent to $M_s \models @_\alpha p$.
    
    Assume that $M_s \models \psi$ is equivalent to $M_s \models @_\alpha p$. Then, by the definition of the semantics, $L^{p_\alpha := \psi} (p) = L(p)$, and therefore $M_s^{p_\alpha := \psi} = M_s$ and $\varphi$ is evaluated in $M_s$. This is equivalent to the conjunct $(\psi \leftrightarrow @_\alpha p) \to \varphi$ in the translation. 

    Finally, assume that $M_s \models \psi$ is not equivalent to $M_s \models @_\alpha p$ and hence $L^{p_\alpha := \psi} (p) \neq L(p)$. This means that whatever the truth-value of $p$ is, we need to flip it at state $\alpha$ and leave the same in other states. To this end, pick any $u \in S$, and we need to show that $M_u \models  \alpha \to \lnot p \land \lnot \alpha \to p $ iff $u \in L^{p_\alpha := \psi} (p)$. Assume that $u \in L^{p_\alpha := \psi} (p)$. If $\{u\} \neq L(\alpha)$, then trivially  $M_u \models \alpha \to \lnot p \land \lnot \alpha \to p$, i.e. the truth-value of $p$ remains the same. If $\{u\} = L(\alpha)$, then, by the semantics of the update modality and the fact that $M_s \models \psi$ is not equivalent to $M_s \models @_\alpha p$, we get that the truth-value of $p$ in $u$ was updated, i.e. $u \not \in L (p)$ and $M_s \models \psi$. Hence $M_u \models \alpha \to \lnot p$ and, therefore, $M_u \models\alpha \to \lnot p \land \lnot \alpha \to p$. Similar reasoning holds for the case of $u \not \in L^{p_\alpha := \psi} (p)$.

    Now, having the translation we apply it recursively from the inside-out, taking the innermost occurrence of a substitution modality and translating it away. The procedure terminates due to formulas of $\mathsf{SLAMB}$ being finite.
}
 \end{proof}

An observant reader may notice, that the presented 
translation of a $\mathsf{SLAMB}$ formula may result in an exponentially larger formula of $\mathsf{HATL}$. 
Or, equivalently, that the initial formula of $\mathsf{SLAMB}$ is exponentially \emph{more succinct} than its translation. Such a blow-up is quite natural for reduction-based translations in $\mathsf{DEL}$ 
\cite{lutz06,french13}. In the next section we will show that despite this, 
model checking $\mathsf{SLAMB}$ is \Ptime-complete.
 
Now we turn to 
$\mathsf{ALAMB}$ and show that the ability to move arrows grants us additional expressivity.

\begin{theorem}
\label{thm:hatlVSalamb}
$\mathsf{HATL} \prec \mathsf{ALAMB}$ and $\mathsf{SLAMB} \prec \mathsf{ALAMB}$.
\end{theorem}

\begin{proof}
The fact that $\mathsf{HATL} \preccurlyeq \mathsf{ALAMB}$ follows trivially as  $\mathsf{HATL}$ is a fragment of $\mathsf{ALAMB}$.
To see that $\mathsf{ALAMB} \not \preccurlyeq \mathsf{HATL}$ recall Example \ref{ex:expressive}. 
It is easy to see that models $M_s$ and $N_s$ cannot be distinguished by any $\mathsf{HATL}$ formula\footnote{\rust{In fact, there is an alternating bisimulation \cite{agotnes07irr} between the models. The discussion of an appropriate notion of bisimulation for $\LAMB$ is outside of the scope of this paper and is left for future work.}}. Indeed, in both models states $s$ and $t$ agree on their corresponding nominals and propositional variables, and, moreover in all states none of the agents can force a transition on their own.
The fact that such a transition requires different 
action profiles by $\{1,2\}$ in different models cannot be captured by formulas of $\mathsf{HATL}$ as they do not have the access to particular actions, rather just to abilities of the agent.

At the same time, as shown in Example \ref{ex:expressive}, the $\mathsf{ALAMB}$ formula $[\alpha \xrightarrow{aa} \alpha] \llb \{1\} \rrb \X \alpha$ holds in $M_s$ and is false in $N_s$, and, therefore, $\mathsf{HATL} \prec \mathsf{ALAMB}$.
The fact that $\mathsf{SLAMB} \prec \mathsf{ALAMB}$ follows by transitivity of the expressivity relation from $\mathsf{HATL} \approx \mathsf{SLAMB}$.
\end{proof}

From Theorem \ref{thm:hatlVSalamb} it 
follows that $\mathsf{HATL} \prec \mathsf{LAMB}$ and $\mathsf{SLAMB} \prec \mathsf{LAMB}$. 
For future work, we leave open the question of whether $\LAMB$ is strictly more expressive than $\mathsf{ALAMB}$, and conjecture that it is indeed the case.
Figure \ref{fig:expressivity} summarises the expressivity results.

\begin{figure}[h!]
\centering
\begin{tikzpicture}[scale=0.8, transform shape]
\node (HATL) at (0,0) {$\mathsf{HATL}$};
\node (SLAMB) at (3,0) {$\mathsf{SLAMB}$};
\node (ALAMB) at (6,0) {$\mathsf{ALAMB}$};
\node (LAMB) at (9,0) {$\mathsf{LAMB}$};

\draw[thick,<->] (HATL) to node[sloped, anchor=center, above] {\footnotesize{Thm. \ref{thm:hatlEQslamb}}} (SLAMB);
\draw[thick,->] (SLAMB) to node[sloped, anchor=center, above] {\footnotesize{Thm. \ref{thm:hatlVSalamb}}} (ALAMB);
\draw[thick,->] (ALAMB) to [bend right] (LAMB);
\draw[thick,->,dashed] (LAMB) to [bend right] node[above] {?}  (ALAMB);

\end{tikzpicture}
\caption{Overview of the expressivity results. An arrow from $\mathsf{L}_1$ to $\mathsf{L}_2$ means $\mathsf{L}_1 \preccurlyeq\mathsf{L}_2$. If there is no symmetric arrow, then $\mathsf{L}_1 \prec \mathsf{L}_2$. This relation is transitive, and we omit transitive arrows in the figure. The dashed arrow with the question mark denotes the open question.}
\label{fig:expressivity}
\end{figure}
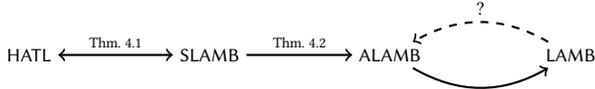

\section{Model checking} 
\label{sec:mc}

The \rust{model checking problem for $\LAMB$} consists in determining, for a CGM $M_s$ and a formula $\varphi$, whether $M_s \models \varphi$. We show 
that despite the increase in expressivity, the complexity of the model checking problem for the full language of $\LAMB$ is \rust{still} \Ptime-complete. 

The complexity of model checking $\ATL$ is known to be \Ptime-complete \cite{alur2002}, and it is easy to see that it remains the same also for $\mathsf{HATL}$. The algorithm for $\ATL$ uses function $Pre (M, C, Q)$ that computes for a given CGM $M$, coalition $C \subseteq Agt$ and a set $Q \subseteq S$, the set of states, from which coalition $C$ can force the outcome to be in one of the $Q$ states. Function $Pre$ can be computed in polynomial time. We can use exactly the same algorithm for computing $Pre$ as for the standard $\ATL$, however in our case we will compute $Pre$ for not only the original model model $M$, but its updated versions as well. 

\begin{breakablealgorithm}
	\caption{An algorithm for model checking $\LAMB$}\label{lambMC} 
	\small
	\begin{algorithmic}[1] 		
		\Procedure{MC}{$M, s, \varphi$}		
        \Case {$\varphi = \alpha$}
            \State{\textbf{return} $s \in L(\alpha)$}
        \EndCase
        \Case {$\varphi = @_\alpha \psi$}
        \If {$L(\alpha) \neq \emptyset$}
            \State{\textbf{return}  $\textsc{MC} (M, L(\alpha), \psi)$}
        \Else
            \State{\textbf{return}  \textit{false}}
        \EndIf
        \EndCase

    
\Case {$\varphi = [\pi] \psi$ with $\pi\in \{p_\alpha:=\psi, \alpha \xrightarrow{A} \beta, \lcirc{\alpha}\}$}
			\State{\textbf{return} $\textsc{MC} (\textsc{Update} (M, s, \pi), s, \psi)$}
		\EndCase
   \EndProcedure

	\end{algorithmic}
\end{breakablealgorithm}

The model checking algorithm for $\LAMB$ (Algorithm \ref{lambMC}) is similar to the one for $\ATL$ when it comes to temporal modalities and Boolean cases, and thus we omit them for brevity (see the full algorithm in the  Appendix). Apart from them, we have hybrid cases $\varphi = \alpha$ and $\varphi = @_\alpha \psi$, and the dynamic case $\varphi = [\pi] \psi$ with $\pi\in \{p_\alpha:=\psi, \alpha \xrightarrow{A} \beta, \lcirc{\alpha}\}$. Regarding the  $\varphi = @_\alpha \psi$, we evaluate $\psi$ at state named $\alpha$, if the state with such a name exists. If the denotation of name $\alpha$ is empty, then $@_\alpha \psi$ is false. 

The dynamic case $\varphi = [\pi] \psi$ is a bit more involved as the algorithm evaluates $\psi$ in a new updated model $M^\pi$. Procedure \textsc{Update} for constructing updated models is captured by Algorithm \ref{update}.

\begin{breakablealgorithm}
	\caption{An algorithm for computing updated models}\label{update} 
	\small
	\begin{algorithmic}[1] 	
    \Procedure{Update}{$M, s, \pi$}
   \Case{$\pi = p_\alpha:=\psi$}
    \If{$L(\alpha) \neq \emptyset$}
        \If{$\textsc{MC} (M,s, \psi)$}
            \State{$L^\pi (p) = L (p) \cup L(\alpha)$}
        \Else
            \State{$L^\pi (p) = L (p) \setminus L(\alpha)$}
        \EndIf
        \State{\textbf{return} $M^\pi = \langle S, \tau, L^\pi \rangle$}
    \Else
        \State{\textbf{return} $M$}
    \EndIf
   \EndCase
   
   \Case{$\pi = \alpha \xrightarrow{A} \beta$}
        \If{$L(\alpha) \neq \emptyset$ and $L(\beta) \neq \emptyset$}
            \State{$\tau^\pi = \tau \setminus \{(L(\alpha), A, \tau(L(\alpha), A))\} \cup \{(L(\alpha), A, L(\beta))\} $}
            \State{\textbf{return} $M^\pi = \langle S, \tau^\pi, L \rangle$}
        \Else
            \State{\textbf{return} $M$}
        \EndIf
   \EndCase

      \Case{$\pi = \lcirc{\alpha}$}
   \If{$L(\alpha) = \emptyset$}
        \State{$S^\pi = S \cup \{t\}$, where $t$ is fresh}
        \State{$\tau^\pi = \tau \cup \{(t, A, t) \mid A \in \Act^\AG\}$}
        \State{$L^\pi = L \cup \{(\alpha, \{t\})\}$}
        \State{\textbf{return} $M^\pi = \langle S^\pi, \tau^\pi, L^\pi \rangle$}
    \Else
        \State{\textbf{return} $M$}
   \EndIf
   \EndCase
   \EndProcedure
	\end{algorithmic}
\end{breakablealgorithm}

In the procedure, an updated model is constructed according to Definition \ref{def:semLAMB}. For the case of substitutions, we first check whether state named $\alpha$ exists, and if it does, we update the valuation function $L$ based on whether $M_s \models \psi$. For arrows $\alpha \xrightarrow{A} \beta$, if both states named $\alpha$ and $\beta$ exist, we substitute in $\tau$ transition $(L(\alpha), A,$ $\tau(L(\alpha), A))$ (i.e. transition from state named $\alpha$ via $A$ to whatever state is assigned according to $\tau(L(\alpha), A)$) by the required transition $(L(\alpha), A, L(\beta))$ (line 13). Finally, to add a new state with name $\alpha$, we first check whether the name is not used, and then extend $S$, $\tau$, and $L$ of the original model accordingly.

Model checking for $\LAMB$ is done then recursively by a combination of procedures \textsc{MC}, which decides whether a given formula is true in a given model, and \textsc{Update}, which computes required updated models. 
\textsc{MC} calls  \textsc{Update} when it needs to perform an update (line 10, Alg. \ref{lambMC}), and  \textsc{Update} calls \textsc{MC} when it needs to compute the valuation of $\psi$ for case $\pi = p_\alpha:=\psi$ (line 4, Alg. \ref{update}).

We run \textsc{MC}($N,t,\psi$) for at most $|\varphi|$ formulas $\psi$ and at most $|\varphi|$ models $N$. Each run, similarly to the algorithm for $\ATL$, is done in polynomial time with respect to $|M|$. 
Hence, procedure \textsc{MC} is used by the model checking algorithm for a polynomial amount of time.   

At the same time, we run \textsc{Update}($N, t, \pi$) for at most $|\varphi|$ models $N$ and at most $|\varphi|$ formulas $\psi$ (the substitution case). The sizes of updated models are bounded by $|\varphi|\cdot|M|$ (the case of adding a new state). Thus, each run of \textsc{Update} takes polynomial time, and hence we spend a polynomial amount of time in the procedure while performing model checking. 

Both procedures, \textsc{MC} and \textsc{Update}, take polynomial time to run, and, therefore, model checking for $\LAMB$ can be done in polynomial time. The lower bound follows straightforwardly from \Ptime-completeness of $\ATL$ model checking.


\begin{theorem}\label{thm:MC}
    The model checking problem for $\LAMB$ is \Ptime-complete.
\end{theorem}

\begin{remark}
To make the language of $\LAMB$ more succinct, we can extend it with constructs $[\pi \cup \rho] \varphi$, meaning ‘whichever update we implement, $\pi$ or $\rho$, $\varphi$ will be true (in both cases)’. The model checking of the resulting logic is \Pspace-complete. Details about the extension and proofs can be found in section \emph{A Note On Succinctness} in the Technical Appendix.
\end{remark}

\section{Related work }
\label{sec:related}


\paragraph*{Strategic Reasoning.}

From the perspective of strategic reasoning, our work is related to the research on rational verification and synthesis. The first is the problem of checking whether a temporal goal is 
satisfied in some (or all) game-theoretic equilibria of a CGM \cite{AbateGHHKNPSW21,GutierrezNPW23}. Rational synthesis consists in the automated construction of such a model \cite{FismanKL10, CFGR16}.  
In this direction, 
\cite{KR2024-44} investigated the problem of finding \emph{incentives} by manipulating the weights of atomic propositions to persuade 
agents to act towards a temporal goal. 

Recent work has also investigated the use of formal methods to verify and synthesize mechanisms for social choice using model checking and satisfiability procedures for variants of Strategy Logic \cite{SLKF_KR21,MittelmannMMP22,MittelmannMMP23}. 
While being able to analyse MAS with respect to complex solution concepts, all these works face high complexity issues. 
In particular, key decision problems for rational verification with temporal specifications are known to be \Dexptime-complete \cite{GutierrezNPW23} and model checking  Strategy Logic is \NONELEMENTARY for memoryfull agents \cite{MogaveroMPV14}. 
\rust{
Compared to these approaches, $\LAMB$ offers relatively high expressivity while maintaining  the \Ptime-completeness of its model checking problem.
}


\rust{
The recently introduced
\textit{obstruction ATL} \cite{catta23,catta24} ($\mathsf{O}\ATL$) allows reasoning 
about agents' strategic abilities 
while being hindered by an external force, called the \textit{Demon}. Being inspired by sabotage modal logic \cite{vanbenthem05}, in this logic  
the Demon is able to disable some transitions and thus impact the strategic abilities of the agents in a system. This is somewhat related to normative updates that we covered in Section \ref{sec:norms} and the \textit{module checking problem} \cite{kupferman2001module,jamroga2015module}, where agents interact with a non-deterministic environment that may inhibit access to certain paths of the computation tree. 

Notice that $\LAMB$ is significantly more general than the presented approaches 
as it allows to not only restrict transitions or access, but also change it in a more nuanced way by redirecting arrows (and, e.g., granting access to a state). Moreover, $\LAMB$ also allows adding \textit{new} states, as well as changing the valuations of propositions. Moreover, updates in $\LAMB$ are explicitly present in the syntax that enables explicit synthesis of model modifications.
}

\rust{
\paragraph*{Nominals.} Nominals are an integral part of \textit{hybrid logic} \cite{ARECES2007821} and is a common tool whenever one needs to refer to particular states on the syntax level. For example, nominals and other hybrid modalities are ubiquitous in the research on \textit{logics for social networks} (see \cite[Chapter 3]{minathesis} for a comprehensive overview). 

In the setting of DEL, tools and methods of hybrid logic have been used, for example, to relax the assumption of common knowledge of agents' names \cite{wang18}, to study the interplay between public announcements and distributed knowledge \cite{HANSEN201133}, and to tackle the information and intentions dynamics in interactive scenarios \cite{Roy2009}. Moreover, nominals were used to provide an axiomatisation of a hybrid variant of \textit{sabotage modal logic} \cite{vanbenthem23}, which extends the standard language of modal logic with constructs $\blacklozenge \varphi$ meaning `after removing some edge in the model, $\varphi$ holds' \cite{vanbenthem05,aucher18}.

Nominals were also used in \textit{linear-} and \textit{branching-time temporal logics} 
to refer to particular points in computation (see, e.g.,  \cite{blackburn99,Goranko2000-GORCTL-2,Franceschet_etal2003,franceschet06,Kara2009,Lange2009,BOZZELLI2010454,Kernberger2020}). 
In the framework of strategic reasoning, \cite{Huang_Meyden2018} 
used some ideas from hybrid logic,
but neither $@_{\alpha}$ nor nominals themselves. Hence, in terms of novelty,   
to the best of our knowledge, the  \textit{Hybrid ATL} ($\mathsf{HATL}$) proposed in this paper is the first attempt to combine nominals with the $\mathsf{ATL}$-style strategic reasoning.

}





\rust{
\paragraph*{The Interplay Between DEL and Strategic Reasoning.}

As we mentioned in the introduction, albeit DEL and various strategic logics being very different formalism, some avenues of DEL research has incorporated ideas from logics for strategic reasoning. Examples include the exploration of \textit{concurrent DEL games} \cite{maubert20}, \textit{alternating-time temporal DEL} \cite{delima14}, \textit{coalitions announcements} \cite{agotnes08,galimullin21b} and other forms of \textit{strategic multi-agent communication} (see, e.g., \cite{agotnes10GAL,galimullin24}).

To the best of our knowledge, DEL updates for CGMs, up until now, were considered only in \cite{galimullin21,galimullin2022action}, where the authors capture granting and revoking actions of singular agents as well as updates based on \textit{action models} \cite{bms22}. Both works are limited to the ne$\mathsf{X}$t-time fragment of $\ATL$ (so-called \textit{coalition logic} \cite{pauly02}). Moreover, they do not support such expressive features of $\LAMB$ as adding \textit{new} states and changing the valuation of propositional variables. Additionally, our arrow-redirecting operators allow for greater flexibility while dealing with agents' strategies. }

\balance 

\section{Discussion \& Conclusion}

We proposed $\LAMB$, a logic for updating CGMs that combines ideas from both the strategy logics tradition ($\ATL$ in our case) and the $\mathsf{DEL}$ tradition. We have argued that $\LAMB$ can be useful for reasoning about a variety of dynamic phenomena in MAS thanks to the modular nature of its primitive update operators. Finally, we have explored the expressivity hierarchy of $\LAMB$ and its fragments, and demonstrated that the model checking problem for $\LAMB$ is \Ptime-complete.   

As we have just scratched the surface of dynamic updates for CGMs, there is a plethora of open questions. One of the immediate ones is to explore the satisfiability problem for $\LAMB$. \rust{
Another one is to assign costs to different types of model changes. In such a way, we will be able to generalise the bounded modification synthesis to the scenarios, where some changes are more costly to implement and thus are less optimal. Moreover, having costs associated with model changes will allow for a direct comparison of (this generalised version of) $\LAMB$ and Obstruction $\ATL$ (see Section \ref{sec:related}). 
We conjecture that in such a scenario, $\LAMB$ will subsume $\mathsf{O}\ATL$.  
}


\rust{
Apart from that, in Section \ref{sec:synth}, we mentioned the exploration of constructive solutions to the synthesis problem, both bounded and unbounded versions, as a promising area of further research. One way to go about it is, perhaps, taking intuitions from the constructive approaches to the $\ATL$ satisfiability \cite{goranko06,walther06,goranko09}. 
Those solutions can be then embedded into some of the existing model checkers, like MCMAS \cite{lomuscio17} and STV \cite{kurpiewski19}, so that if a model checker returns FALSE for a given model and property $\varphi$, the tool automatically constructs an update that can fix the model so that it satisfies $\varphi$.
}

In a more general setting of CGM updates, one can also consider modifications that cannot be captured by $\LAMB$. For example, we can explore the effects of granting or revoking actions to/from certain agents, changing the number of agents, or any combinations thereof with the $\LAMB$ updates. \rust{As mentioned in Section \ref{sec:related},} some preliminary work on changing the actions available to agents has been done for the ne$\mathsf{X}$t-time fragment of $\ATL$ \cite{galimullin21,galimullin2022action}.

Finally, as $\LAMB$ is based on $\ATL$, we find it particularly interesting to consider more expressive base languages, like, $\ATL^\star$ or variations of strategy logic $\mathsf{SL}$ \cite{MogaveroMPV14}. Of particular interest is the \textit{simple goal fragment of} $\mathsf{SL}$ \cite{belardinelli19}, which is strictly more expressive than $\ATL$ and yet allows for  \Ptime-time model checking. 
\rust{
We would also like to consider the ideas of model updates in the setting of 
STIT logics \cite{horty01,broersen15}.  
Additionally, we believe that ideas from Separation Logics
\cite{Reynolds02,DemriD15,DemriF19}, which were proposed to verify programs with mutable data structures, could also provide insights on how to reason about separation and composition-based modifications of MAS.  
}

\begin{acks}
\ifarxiv This project has received funding from the European Union’s Horizon 2020 research and innovation programme under the Marie Skłodowska-Curie grant agreement No 101105549. 
\fi
\end{acks}



\bibliographystyle{ACM-Reference-Format} 
\bibliography{REF}

\clearpage
\appendix
\section*{Technical Appendix}
\subsection*{Model checking $\LAMB$}
Full model checking algorithm for $\LAMB$.
\begin{breakablealgorithm}
	\caption{An algorithm for model checking $\LAMB$}
	\small
	\begin{algorithmic}[1] 		
		\Procedure{MC}{$M, s, \varphi$}		
      \Case {$\varphi = p$}
            \State{\textbf{return} $s \in L(p)$}
        \EndCase
        \Case {$\varphi = \alpha$}
            \State{\textbf{return} $s \in L(\alpha)$}
        \EndCase
        \Case {$\varphi = @_\alpha \psi$}
        \If {$L(\alpha) \neq \emptyset$}
            \State{\textbf{return}  $\textsc{MC} (M, L(\alpha), \psi)$}
        \Else
            \State{\textbf{return}  \textit{false}}
        \EndIf
        \EndCase
        \Case {$\varphi = \lnot \psi$}
            \State{\textbf{return}  not $\textsc{MC} (M, s, \psi)$}
        \EndCase
       \Case {$\varphi = \psi \land \chi$}
            \State{\textbf{return} $\textsc{MC} (M,s,\psi)$ and  $\textsc{MC} (M,s,\chi)$}
        \EndCase
      \Case {$\varphi = \llb C \rrb \X \psi$}
             \State{\textbf{return} $s \in Pre(M, C, \{t \in S | \textsc{MC}(M, t, \psi)\})$}
        \EndCase
    \Case{$\varphi = \llb C \rrb \psi \U  \varphi$}
    \State{$X:= \emptyset$ and $Y:= \{t \in S \mid \textsc{MC} (M,t, \varphi)$\}}
    \While{$Y \neq X$}
        \State{$X:= Y$}
        \State{$Y := \{t \in S \mid \textsc{MC} (M,t, \varphi)\} \cup (Pre(M, C, X) \cap \{t \in S \mid \textsc{MC} (M,t, \psi)\})$}
    \EndWhile
    \State{\textbf{return} $s \in X$}
    \EndCase

    \Case{$\varphi = \llb C \rrb \psi \R \varphi$}
    \State{$X:= S$ and $Y := \{t \in S \mid \textsc{MC} (M,t, \varphi)\}$}
    \While{$Y \neq X$}
        \State{$X:= Y$}
        \State{$Y := \{t \in S \mid \textsc{MC} (M,t, \varphi)\} \cap (Pre(M, C, X) \cup \{t \in S \mid \textsc{MC} (M,t, \psi)\})$}
    \EndWhile
    \State{\textbf{return} $s \in X$}
    \EndCase
    
\Case {$\varphi = [\pi] \psi$ with $\pi\in \{p_\alpha:=\psi, \alpha \xrightarrow{A} \beta, \lcirc{\alpha}\}$}
			\State{\textbf{return} $\textsc{MC} (\textsc{Update} (M, s, \pi), s, \psi)$}
		\EndCase
   \EndProcedure

	\end{algorithmic}
\end{breakablealgorithm}

\subsection*{A Note On Succinctness}
When we discussed normative updates on CGMs in Section \emph{Dynamic MAS through the lens of $\LAMB$}, we, as a use case, considered sanctioning norms $\mathsf{SN}$'s and their effects on a given system. For example, we may want to check whether the systems is \textit{compliant} with some set of norms, i.e. whether none of the normative updates violate some desired property $\varphi$ (e.g. a safety requirement).  In other words, for a given model $M_s$ and a set of norms $\mathcal{N} = \{\mathsf{SN}_1, ..., \mathsf{SN}_n\}$, we can explicitly check whether $M_s$ satisfies $\varphi$ after each $\mathsf{SN}$. As described in the normative updates section, we can model the effects of norms $\mathcal{N}$ in the language of $\LAMB$. Hence, the compliance just described can be expressed by formula $\bigwedge_{i \in \{1,...,n\}} [\mathit{SN}_i]\varphi$, where $\mathit{SN}_i$ are the translation of norms from $\mathcal{N}$'s into $\LAMB$ updates.

The reader might have noticed, however, that such a formula is not quite succinct, as it may have exponentially many repetitions, especially in the case of nested update operators. This happens, for example, if we have several sets of norms that we would like to implement consecutively.  For instance, if the desired formula looks like $\bigwedge_{i \in \{1, 2\}} \bigwedge_{j \in \{1,2,3\}}[\mathit{SN}_i^a][\mathit{SN}_j^b]\varphi$, writing it out in full yields us 
\begin{gather*}
[\mathit{SN}_1^a]([\mathit{SN}_1^b]\varphi \land [\mathit{SN}_2^b]\varphi \land [\mathit{SN}_3^b]\varphi) \land \\
[\mathit{SN}_2^a]([\mathit{SN}_1^b]\varphi \land [\mathit{SN}_2^b]\varphi \land [\mathit{SN}_3^b]\varphi).
\end{gather*}
The reader may think of a fleet of warehouse robots that we want to govern on two levels: first set of norms would regulate the traffic laws in the warehouse to avoid collisions, and the second set would regulate strategic objectives of robots \textit{given the implemented traffic laws}.

To deal with such a blow-up, we can consider a variant of $\LAMB$, called $\LAMB^\cup$, that extends the former with constructs $[\pi \cup \rho]\varphi$ with the intended meaning `whichever update we implement, $\pi$ or $\rho$, $\varphi$ will be true (in both cases)'. Such a union of actions, inherited from Propositional Dynamic Logic ($\mathsf{PDL}$) \cite{fischer79}, is also quite used in $\mathsf{DEL}$ (see, e.g., \cite{bms22,aucher13}). Corresponding fragments of $\LAMB^\cup$ are denoted as  $\mathsf{SLAMB}^\cup$ and $\mathsf{ALAMB}^\cup$.  

Semantics of the union operator is defined as 
  \begin{alignat*}{3}
        &M_s \models [\pi \cup \rho] \varphi && \text{ iff } &&M_s \models [\pi] \varphi \text{ and } M_s \models [\rho] \varphi
\end{alignat*} 
It is easy to see now that $[\pi \cup \rho] \varphi \leftrightarrow [\pi] \varphi \land [\rho] \varphi$ is a valid formula (i.e. it is true on all CGMs $M_s$), and hence we can translate every formula $\LAMB^\cup$ into an equivalent formula of $\LAMB$ (the same for $\mathsf{SLAMB}^\cup$ and $\mathsf{ALAMB}^\cup$). Therefore, expressivity-wise, $\LAMB^\cup \approx \LAMB$, $\mathsf{SLAMB}^\cup \approx \mathsf{SLAMB}$, and $\mathsf{ALAMB}^\cup \approx \mathsf{ALAMB}$.

As noted above, however, that $\LAMB^\cup$ allows us to express, for example, the compliance property for normative updates exponentially \textit{more succinct}. Recalling our example of \begin{gather*}
[\mathit{SN}_1^a]([\mathit{SN}_1^b]\varphi \land [\mathit{SN}_2^b]\varphi \land [\mathit{SN}_3^b]\varphi) \land \\
[\mathit{SN}_2^a]([\mathit{SN}_1^b]\varphi \land [\mathit{SN}_2^b]\varphi \land [\mathit{SN}_3^b]\varphi),
\end{gather*} in the syntax of $\LAMB^\cup$, we can succinctly write an equivalent $[\mathit{SN}_1^a \cup \mathit{SN}_2^a][\mathit{SN}_1^b \cup \mathit{SN}_2^b \cup \mathit{SN}_3^b]\varphi$.

This increased succinctness, however, comes at a price. The complexity of the model checking problem for $\LAMB^\cup$ (and its fragments) jumps all the way to \Pspace-complete. Here we present a sketch of the argument.

\begin{theorem}
\label{lambcupMC}
    The model checking problem for $\LAMB^\cup$, $\mathsf{SLAMB}^\cup$ and $\mathsf{ALAMB}^\cup$ is \Pspace-complete.
\end{theorem}

\begin{proof}
The model checking algorithm for  $\LAMB^\cup$ extends the one for $\LAMB$ with one additional case (Algorithm \ref{pspaceMC}).

\begin{breakablealgorithm}
	\caption{An algorithm for model checking $\LAMB^\cup$}\label{pspaceMC} 
 \footnotesize
	\begin{algorithmic}[1] 		
		\Procedure{MC$^\cup$}{$M, s, \varphi$}

\Case {$\varphi = [\pi \cup \rho] \psi $ }
			\State{\textbf{return} $\textsc{MC$^\cup$} (M,s, [\pi]\psi)$ and $\textsc{MC$^\cup$} (M,s, [\rho]\psi)$}
		\EndCase
   \EndProcedure

	\end{algorithmic}
\end{breakablealgorithm}

New models require space that is bounded by $|M|\cdot|\varphi|$.
There are at most $|\varphi|$ symbols in $\varphi$, and hence the total memory space required by the algorithm is bounded by $|M| \cdot |\varphi|^2$.


\Pspace-hardness is shown via the reduction from the satisfiability of quantified Boolean formulas (QBFs). W.l.o.g., we assume that there are no free variables in QBFs, and that each variable is quantified only once. Given a QBF $\Psi := Q_1 x_1 ... Q_n x_n \Phi (x_1, ..., x_n)$ with $Q_i \in \{\exists, \forall\}$, we construct a CGM $M_s$ and a formula $\psi \in \LAMB^\cup$ (both of polynomial size with respect to $\Psi$), such that $\Psi$ is true if and only if $M_s \models \psi$.

CGM $M = \lb S, \tau, L\rb$ is constructed over one agent and $\Act = \{a_1, ..., a_n\}$, where $S = \{t, s_1, ..., s_{n}\}$, $\tau(s_i, a_j) = s_i$ for all $s_i \in S$ and $a_j \in \Act$, $L(p_i) = \{s_i\}$ with $p_i \in P$, $L(x_i) = \{s_i\}$ with $x_i \in \Nom$, and $L(x_t) = \{t\}$. Intuitively, the CGM consists of $n+1$ states with reflexive loops for each action of the agent. Propositional variables $p_i$ and nominals $x_i$ are true only in the $i$th state. 

The translation of $\Psi$ into a formula of $\LAMB$ is done recursively as follows:
\begin{align*}
    \psi_0 &:= \Phi(\llb 1 \rrb p_1, ..., \llb 1 \rrb p_n) \\
    \psi_k &:=
    \begin{cases}
        [x_t \xrightarrow{a_k} x_k \cup x_t \xrightarrow{a_k} x_t] \psi_{k-1} &\text{if } Q_k = \forall\\ 
        \lnot [x_t \xrightarrow{a_k} x_k \cup x_t \xrightarrow{a_k} x_t] \lnot \psi_{k-1} &\text{if } Q_k = \exists\\
    \end{cases}\\
\psi &:= \psi_n
\end{align*}
Now we argue that 
\[Q_1 x_1 ... Q_n x_n \Phi (x_1, ..., x_n) \text{ is satisfiable iff } M_t \models \psi.\]
Redirection of a transition labelled with an action profile $a_k$ from state $t$ to state $s_k$ models setting variable $x_k$ to 1. Similarly, if the transition brings us from $t$ back to $t$, then variable $x_k$ is set to 0. Since the transition function is deterministic, the choice of truth values is unambiguous. 

To model quantifiers, we use the union operator. The universal quantifier $\forall x_k$ is translated into $[x_t \xrightarrow{a_k} x_k \cup x_t \xrightarrow{a_k} x_t] \psi_{k-1}$ meaning that no matter what the chosen truth-value of $x_k$ is, subformula $\psi_{k-1}$ is true. Similarly, for the existential quantifier $\exists x_k$, we state that there is a choice of the truth-value of $x_k$ such that $\psi_{k-1}$ is true. Once the valuation of all $x_k$ has been set, the evaluation of the Boolean subformula of the QBF corresponds to the reachability of $s_k$'s via an $a_k$-labelled transitions. 

In the hardness proof just presented 
 we used only arrow change operators. Hence, we have shown that $\mathsf{ALAMB}^\cup$, and therefore $\LAMB^\cup$, have \Pspace-complete model checking problems.  

Now, we turn to the the case of substitutions. Let $$\Psi := Q_1 x_1 ... Q_n x_n \Phi (x_1, ..., x_n)$$ be a QBF. We construct a CGM $M = \lb S, \tau, L\rb$ over one agent and one action $\Act = \{a\}$, where $S = \{s\}$, $\tau(s , a) = s$, $L(p^i) = \{s\}$ with $p^1, ..., p^n \in P$, $L(\alpha) = \{s\}$ with $\alpha \in \Nom$. Intuitively, the CGM consists of a single state with a single reflexive loop. Propositional variables $p^i$'s, corresponding to QBF variables $x_i$'s, and nominal $\alpha$ is true in the single state. 

The translation of $\Psi$ into a formula of $\LAMB$ is as follows:

\begin{align*}
    \psi_0 &:= \Phi(p^1, ..., p^n) \\
    \psi_k &:=
    \begin{cases}
        [p^k_\alpha:= \top \cup p^k_\alpha := \bot] \psi_{k-1} &\text{if } Q_k = \forall\\ 
        \lnot [p^k_\alpha:= \top \cup p^k_\alpha := \bot] \lnot \psi_{k-1} &\text{if } Q_k = \exists\\
    \end{cases}\\
\psi &:= \psi_n
\end{align*}

It is relatively straightforward to see that formula $\psi$ explicitly models quantifiers $Q_i$ for variables $p^i$. Then subformula  $\Phi(p^1, ..., p^n)$ is trivially evaluated in the single state of the model. 
\end{proof}

\end{document}